\documentclass[11pt]{article}
\usepackage[utf8]{inputenc}
\usepackage[margin=2.5cm]{geometry}
\usepackage{amsmath,amsthm,amssymb}
\usepackage[numbers]{natbib}
\usepackage{MnSymbol}
\usepackage[ruled]{algorithm2e}

\newtheorem{theorem}{Theorem}
\newtheorem*{theorem*}{Theorem}
\newtheorem{definition}{Definition}
\newtheorem{lemma}{Lemma}
\newtheorem*{lemma*}{Lemma}
\newtheorem{conjecture}{Conjecture}
\newtheorem*{conjecture*}{Conjecture}

\newtheorem*{definition*}{Definition}
\newtheorem{proposition}{Proposition}
\newtheorem{remark}{Remark}
\newtheorem{example}{Example}
\newtheorem{corollary}{Corollary}

\newcommand{\cV}{\mathcal{V}}
\newcommand{\cU}{\mathcal{U}}
\newcommand{\cF}{\mathcal{F}}
\newcommand{\cG}{\mathcal{G}}
\newcommand{\bN}{\mathbb{N}}
\newcommand{\bR}{\mathbb{R}}

\newcommand{\bP}{\mathbb{P}}

\newcommand{\alloc}{S}
\newcommand{\alloci}{S_i}
\newcommand{\allocs}{{\bf S}}

\newcommand{\val}{v}
\newcommand{\vali}{v_i}

\newcommand{\uv}{\underline{v}}
\newcommand{\uM}{\underline{M}}

\newcommand{\thre}[0]{\tau}

\title{Fair Division via Quantile Shares\thanks{This project has received funding from the European Research Council (ERC) under the European Union’s Horizon 2020 research and innovation program (grant agreement No. 866132), by an Amazon Research Award, and by the NSF-BSF (grant number 2020788). V.V. Narayan thanks S. Mauras and D. Mohan for discussions.}}
\author{Yakov Babichenko\thanks{Technion -- Israel Institute of Technology. Email: \texttt{yakovbab@technion.ac.il}} \and Michal Feldman\thanks{Tel Aviv University. Email: \texttt{mfeldman@tauex.tau.ac.il}} \and Ron Holzman\thanks{Technion -- Israel Institute of Technology. Email: \texttt{holzman@technion.ac.il}} \and Vishnu V. Narayan\thanks{Tel Aviv University. Email: \texttt{narayanv@tauex.tau.ac.il}}}

\date{\today}

\begin{document}

\maketitle

\begin{abstract}
We consider the problem of fair division, where a set of indivisible goods should be distributed fairly among a set of agents with combinatorial valuations. 
To capture fairness, we adopt the notion of {\em shares}, where each agent is entitled to a {\em fair share}, based on some fairness criterion, and an allocation is considered fair if the value of every agent (weakly) exceeds her fair share.
A share-based notion is considered {\em universally feasible} if it admits a fair allocation for every profile of monotone valuations.
A major question arises: is there a non-trivial share-based notion that is universally feasible?
The most well-known share-based notions, namely proportionality and maximin share, are not universally feasible, nor are any constant approximations of them.

We propose a novel share notion, where an agent assesses the fairness of a bundle by comparing it to her valuation in a random allocation. In this framework, a bundle is considered {\em $q$-quantile fair}, for $q\in[0,1]$, if it is at least as good as a bundle obtained in a uniformly random allocation with probability at least $q$.
Our main question is whether there exists a constant value of $q$ for which the $q$-quantile share is universally feasible.

Our main result establishes a strong connection between the feasibility of quantile shares and the classical Erd\H{o}s Matching Conjecture. Specifically, we show that if a version of this conjecture is true, then the $\frac{1}{2e}$-quantile share is universally feasible.
Furthermore, we provide unconditional feasibility results for additive, unit-demand and matroid-rank valuations for constant values of $q$.
Finally, we discuss the implications of our results for other share notions.
\end{abstract}

\section{Introduction}

Fair division, the problem of allocating resources in a fair manner, has emerged as a prominent and crucial area of research that has attracted considerable attention in the literature. This challenging problem arises in various practical applications, ranging from classical examples like the division of inherited estates, international border settlements, and the allocation of public resources and government spending, to modern applications such as assigning seats in college courses, allocating computational resources, distributing the electromagnetic spectrum, and managing airport traffic.

\paragraph{Fair division model.} We consider the problem of allocating a set of indivisible goods $[m]=\{1,\ldots,m\}$ among $n$ agents. Every agent $i \in [n]$  has a valuation function $\vali:2^{[m]} \to \bR_+$ that assigns a real value to every bundle of goods. The function $\vali$ is monotone, 
namely $\vali(S')\leq \vali(S)$ for all $S' \subseteq S \subseteq [m]$. The class of all monotone valuation functions will be denoted by $\cV$.
An allocation is a partition of the set of goods among the agents; it is denoted by $\allocs=(\alloc_1,...,\alloc_n)$, where $\alloc_i \subseteq [m]$ denotes the bundle allocated to agent $i \in [n]$, and $\alloc_1\cupdot \alloc_2 \cupdot ... \cupdot \alloc_n = [m]$.
The ultimate goal is to find a {\em fair} allocation, according to some natural notion of fairness. 

Over the years, different notions of fair allocation have been introduced, capturing various interpretations of fairness.
Some examples include envy-freeness~\cite{foley1966resource} and its variants (e.g., EF1, EFX)~\cite{budish2011combinatorial, ckmpsw2019}. 
Others consider egalitarian objectives such as maximizing the utility of the worst-off agent~\cite{bansal2006santa}, or maximizing the product of the agent utilities, known as the Nash welfare objective~\cite{nash1950bargaining,kaneko1979nash,ckmpsw2019}. 
Another approach, which we adopt in this paper, is based on the notion of {\em shares}. Examples include the proportional share~\cite{Ste48} and the maximin share~\cite{budish2011combinatorial}.

\paragraph{The notion of a ``share".}
Following the terminology of \cite{babaioff2022fair}, a \emph{share} $\thre=\thre(\vali, n)$ is a function that maps a pair $(\vali, n)$ to a real value, with the interpretation that any allocation among the $n$ agents that gives agent $i$ a bundle of value at least $\thre(\vali, n)$ is acceptable by agent $i$, thus considered fair towards agent $i$.
Consequently, an allocation $\allocs$ is said to be fair towards agent $i$ if $\vali(\alloci) \geq \thre$.
An allocation $\allocs$ is said to be fair if it is fair towards all agents $i \in [n]$. 
For a definition of a share to be meaningful with respect to some valuation class $\cU \subseteq \cV$, it should be \emph{feasible for the class $\cU$}.
\begin{definition}[Feasible share]
    A share $\thre$ is {\em feasible} for the valuation class $\cU$ if for every $v_1,...,v_n\in \cU$ there exists a fair allocation $\allocs=(S_1,...,S_n)$, namely, an allocation $\allocs$ such that $v_i(S_i)\geq \thre(v_i,n)$ for every $i \in [n]$.
\end{definition}

A share that is feasible for all monotone valuations is termed {\em universally feasible}.

\begin{definition}[Universal feasibility]
   A share $\thre$ is {\em universally feasible} if it is feasible for the class of all monotone valuations. 
\end{definition}

To the best of our knowledge, none of the notions of shares in the literature is universally feasible (except for trivial cases). 
For example, the maximin share (MMS) of an agent is the value the agent obtains by partitioning the goods into $n$ bundles, at her choice, and receiving the worst one among them \cite{budish2011combinatorial}. The maximin share is known to be infeasible, even for additive valuations \cite{kurokawa2018fair, feige2021tight}, but a constant fraction of MMS is feasible for additive valuations \cite{kurokawa2018fair, ghodsi2018fair, garg2020improved,  feige2022improved, akrami2023simplification}, as well as submodular and XOS valuations \cite{ghodsi2022fair}. However, for general valuations, no constant fraction of MMS is feasible (see Example~\ref{ex:mms-infeasible} in Appendix \ref{ap:max-min}).

As another example, the proportional share of an agent is defined by $\thre(\vali,n)=\vali([m])/n$, namely, it is a $1/n$ fraction of agent $i$'s valuation for the grand bundle. The proportional share is infeasible for any constant fraction. This can be easily seen by considering a single good scenario.

\vspace{0.1in}
These examples motivate the following question: \emph{Does there exist a natural share notion that is universally feasible? (i.e., feasible for all monotone valuations)}

\paragraph{Quantile shares.} 
In what follows we define our new notion of shares, termed quantile shares. 
According to this notion, an agent deems a bundle fair or unfair, based on how it compares to a uniformly random bundle; namely, a random bundle $X_i\in \Delta(2^{[m]})$ that contains every good with probability $\frac{1}{n}$, independently across all goods. 
(Equivalently, this is the random allocation of an agent when choosing an allocation uniformly at random among all $n^m$ possible allocations.) 
Specifically, a bundle $T$ is said to be $q$-quantile fair towards agent $i$ (or in short, $q$-fair) if the probability that $T$ is (weakly) better for $i$ than the random bundle $X_i$, is at least $q$. 

A formal definition follows.
Recall that the $q$-quantile of a real-valued distribution with CDF $F$ is defined by $\sup\{t\in \mathbb{R} \mid F(t)< q\}$.

\begin{definition}[$q$-quantile share, $q$-fair]
The \emph{$q$-quantile share $\thre_q(\vali,n)$} is the $q$-quantile of the distribution $v_i(X_i)\in \Delta(\mathbb{R}_+)$.
A bundle $T$ is said to be \emph{$q$-quantile fair towards agent $i$} (or in short, {\em $q$-fair}) if $\vali(T) \geq \thre_q(\vali,n)$.\footnote{Interestingly, if instead of considering quantiles of a distribution $\vali(X_i)$ we consider its expectation, then for an additive valuation $v_i$ we get precisely the definition of the proportional share.}
\end{definition}

Note that quantile shares are ordinal in nature. 
Indeed, to determine whether a particular bundle is fair for an agent, her ordinal preferences over the bundles suffice; no cardinal information is required.
An alternative motivation for this notion is introduced below in Section \ref{sec:veto}.

\subsection{Our Results}
We find an interesting connection between the feasibility of $q$-quantile shares and the famous Erd\H{o}s Matching Conjecture. 
Roughly speaking, we show that if the Erd\H{o}s Matching Conjecture is true, then the $\frac{1}{2e}$-quantile share is universally feasible. {\em To the best of our knowledge, this is the first non-trivial notion of shares that is universally feasible.} 
More specifically, we show the following:

\vspace{0.1in}
\noindent {\bf Theorem:} (see Theorem~\ref{theo:id})
If the Erd\H{o}s Matching Conjecture is true (even for a specified special case), then for every $n\in \bN$,  the $\frac{1}{2e}$-quantile share is feasible for any profile of identical (across agents) monotone valuation. 
\vspace{0.1in}

\vspace{0.1in}
\noindent {\bf Theorem:} (see Theorem~\ref{theo:gen})
If the Rainbow Erd\H{o}s Matching Conjecture is true (even for a specified special case), then for every $n\in \bN$, the $\frac{1}{2e}$-quantile share is universally feasible.
\vspace{0.1in}

These positive results are tight up to a factor of 2; we show that the $\frac{1}{e}$-quantile share is infeasible (see Proposition \ref{pro:1/e}). 

We then turn to unconditional feasibility results 
for special cases. For additive valuations, we show the following:

\vspace{0.1in}
\noindent {\bf Proposition:} (see Proposition~\ref{pro:add})
For every profile of additive valuations, the $q$-quantile share is feasible for every $q<\frac{0.14}{e}$, as $n \to \infty$.
\vspace{0.1in}

For unit-demand and matroid-rank valuations, we identify the critical value of $q=1/e$ as the switching point from feasibility to infeasibility. 

\vspace{0.1in}
\noindent {\bf Proposition:} (see Propositions~\ref{pro:ud} and~\ref{pro:mrf})
For every profile of unit-demand valuations or of matroid-rank valuations, the $q$-quantile share is feasible for every $q < \frac{1}{e}$ 
and is infeasible for $q>\frac{1}{e}$ for sufficiently large $n$.
\vspace{0.1in}

A few remarks are in order. 

First, for small values of $n$ and $m$, one can verify by exhaustive search that for general monotone valuations, the threshold for feasibility is exactly $(1-\frac{1}{n})^{n-1}$, which approaches $\frac{1}{e}$ as $n$ grows. In Section~\ref{subsec:simulation}, we discuss our use of computer simulation towards discovering the critical threshold for $n \in \{3,4\}$ and small values of $m$. 

Second, it is interesting to note that for the case of identical valuations, the above feasibility results imply associated lower bounds 
on the quantiles of the maximin share and the proportional share (see Section~\ref{subsec:comparison}). 

Third, unlike other share notions (e.g., maximin share), quantile shares can essentially be efficiently computed.
Moreover, for additive, unit-demand and matroid-rank valuations, our existence results suggest constructive algorithms (see Section~\ref{subsec:computation}). 

Finally, in Section~\ref{subsec:bads} we show that the feasibility of the $q$-quantile share for constant $q$ does not extend to the allocation of chores (bads).

\subsection{Quantile Shares as Vetoes}\label{sec:veto}
We next provide an alternative interpretation of quantile shares. Consider a scenario where an allocator is unaware of the agent valuations, and every agent reports a \emph{veto list} of allocations they deem unacceptable. The allocator's goal is to come up with an allocation that is not contained in the collective veto lists submitted by the agents.

Fairness here is captured by the fact that each agent is limited to the same size for their veto list. Formally, let $b$ denote the maximal size of the veto list submitted by every agent, and let $L_i$ be the veto list of agent $i$, where $|L_i|\leq b$. The following natural question arises: How large can the parameter $b$ be while ensuring the existence of an allocation $\allocs \notin \bigcup_i L_i$ for all possible list reports?

Clearly, this question is meaningful only when restricting attention to veto lists that satisfy a monotonicity condition. (In the absence of this restriction, one can easily see that $n^m/n$ is a tight threshold. Indeed, since the total number of allocations is $n^m$, then if $b\geq n^m/n$, then the agent lists might cover the entire set of possible allocations, and if $b < n^m/n$, then there must exist an allocation that does not belong to the union of all veto lists.) 

\begin{definition}
    A veto list $L_i$ is \emph{monotonicity-consistent}\footnote{Note that monotonicity-consistency does not allow an agent to interfere in the allocation of others. This is captured by the requirement for the case $S_i=S'_i$.} if $$(S_i,S_{-i})=\allocs\in L_i \Rightarrow (S'_i,S'_{-i})=\allocs'\in L_i \text{ for every } S'_i \subseteq S_i \text{ and every } S_{-i}, S'_{-i}.$$
\end{definition}

The question then becomes: How large can $b$ be while ensuring the existence of an allocation $\allocs \notin \bigcup_i L_i$ for all possible monotonicity-consistent list reports? This question turns out to be equivalent to the feasibility of $q$-quantile shares, either for the full class of monotone valuations $\cV$ or for its subclass $\cV_{01}$ of monotone $0/1$-valuations $u_i: 2^{[m]} \to \{0,1\}$. This is cast in the following proposition.

\begin{proposition} \label{pro:equiv}
The following three statements are equivalent:
\begin{enumerate}
    \item For all monotonicity-consistent lists $L_1,...,L_n$ of size at most $b$ there exists $\allocs \notin \bigcup_i L_i$.
    \item For $q=(b+1)/n^m$ the $q$-quantile share is universally feasible.
    \item For $q=(b+1)/n^m$ the $q$-quantile share is feasible for all monotone $0/1$-valuations.
\end{enumerate}
\end{proposition}

\begin{proof}
    (2 $\Rightarrow$ 3) is trivial. Thus, it remains to prove that (1 $\Rightarrow$ 2) and (3 $\Rightarrow$ 1). 
    
    We first show that (1 $\Rightarrow$ 2). If the $q$-quantile share is infeasible for the profile of monotone valuations $v_1,...,v_n$, then let $L_i=\{\allocs:v_i(S_i)<\tau_q(v_i,n)\}$. Note that $|L_i|\leq b$ because we have at most $b$ allocations whose value for agent $i$ is strictly worse than the $(b+1)$'th worst value. The fact that $\tau_q$ is infeasible implies that for every allocation $\allocs=(S_1,...,S_n)$ there exists an agent $i$ for whom $v_i(S_i)<\tau_q(v_i,n)$; namely $\allocs \in L_i$.  
    
    We next show that (3 $\Rightarrow$ 1). Let $L_1,...,L_n$ be monotonicity-consistent veto lists of size at most $b$. Note that whether or not $\allocs \in L_i$ depends only on $S_i$, and this dependence is monotone. Hence, for each $i \in [n]$, we can define a monotone $0/1$-valuation $u_i$ by: $u_i(S_i)=0$ if $\allocs \in L_i$ and $u_i(S_i)=1$ if $\allocs \notin L_i$. Then in a random allocation we have $u_i(X_i)=0$ with probability at most $b/n^m$ and $u_i(X_i)=1$ with the remaining probability. Since $b/n^m < q$, the $q$-quantile of $u_i(X_i)$ is $1$. The feasibility of the $q$-quantile share implies the existence of an allocation $\allocs=(S_1,...,S_n)$ such that $u_i(S_i)=1$ for all $i\in [n]$, or equivalently $\allocs \notin \bigcup_i L_i$.
\end{proof}

Let us now revisit our question with Proposition~\ref{pro:equiv} in hand.
Recall that in the absence of the monotonicity restriction on the veto lists, feasibility can only be maintained if $b<n^m/n$, that is, the fraction of vetoed allocations shrinks as the number of agents grows.
In contrast, our main results imply that under monotonicity-consistent lists, the size of the lists can be as large as $b=\lfloor n^m/(2e) \rfloor$. 
Namely, the fraction of vetoed allocations can be constant, independent of the number of agents and the valuation profiles, while still ensuring feasibility.

\subsection{Additional Related Literature}
Broadly, the two main types of fairness notions are \textit{share-based} fairness and \textit{envy-based} fairness. The notion of a fair share has remained central to the fair division literature ever since its formal study commenced with the work of Banach, Knaster and Steinhaus~\cite{Ste48} in the 1940s. They worked on the problem of fairly dividing a heterogeneous divisible good (i.e., \textit{cake-cutting}), and devised a procedure to attain the \textit{proportional share}, in which each of the $n$ agents gets a bundle (or piece of the cake) of value at least a $\frac{1}{n}$-fraction of their value for the set of all items (or the entire cake). Several notable papers about proportional cake division have since been published~(\cite{DS61,EP84,EP06}).

However, with indivisible items it is easy to see that proportionality is unattainable: consider the case of two agents and a single good. In this setting, Budish~\cite{budish2011combinatorial} defined a new share-based fairness notion called the maximin share (MMS). For additive valuations, the MMS is weakly smaller than the proportional share; however, Budish~\cite{budish2011combinatorial} left open the problem of whether an MMS allocation (one in which every agent receives a bundle of value at least its maximin share) always exists. Procaccia and Wang~\cite{PW14} show that the MMS is not always feasible even for additive valuations, thereby initiating a line of research into the existence of feasible approximations and relaxations of the MMS guarantee (\cite{PW14,GHS18,AG23}).

More recently, Babaioff and Feige \cite{babaioff2022fair} formally define the general notion of shares and focus on several desirable properties of shares. One of these properties is \emph{self-maximization} which, roughly speaking, incentivizes agents to report their valuation truthfully under a worst-case fair allocation. While the maximin share is itself infeasible, it is self-maximizing. By contrast, while some multiplicative approximations of the maximin share are known to be feasible, no such approximation is self-maximizing~\cite{babaioff2022fair}. It is easy to see that our notion of quantile shares is self-maximizing. Another desirable property is being \emph{undominated}: it should be impossible to promise more value to the agents and still maintain feasibility. We notice that in the class of unit-demand valuations, where we exactly determine the critical value $q$ for feasibility, the corresponding $q$-quantile share is undominated. 

The idea of measuring the satisfaction of agents via quantiles is not new. In the case of two agents the criterion of exceeding the median quantile has been considered in \cite{de2014selection}. Exceeding a general quantile for two agent allocation problems has been considered in \cite{meir2018bidding}.

For envy-based notions of fairness, the notion of \emph{envy-freeness up to one good (EF1)}  is, to the best of our knowledge, the only notion whose feasibility has been proved for general valuation functions (see \cite{budish2011combinatorial}). This notion of fairness relaxes the envy-freeness requirement by allowing an agent to remove one good from the bundle of the opponent before examining whether she envies her. Stronger envy-based fairness notions, such as \emph{envy-freeness up to any good (EFX)}, remain open even for additive valuations.

\section{Preliminary Observations} 

As a warm-up, we present several simple (positive and negative) results on the feasibility of $q$-quantile shares. First, we use the union bound to obtain the following feasibility result for the class $\cV$ of all monotone valuations.

\begin{proposition}\label{pro:un}
    For every $n,m\in \mathbb{N}$ the $1/n$-quantile share is universally feasible.
\end{proposition}

\begin{proof}
    By Proposition~\ref{pro:equiv}, it suffices to show that given any veto lists $L_1,...,L_n$, each of size strictly less than $n^m/n=n^{m-1}$, there exists an allocation $\allocs \notin \bigcup_i L_i$ This follows from the union bound.
\end{proof}

A straightforward infeasibility result is implied by a scenario in which $n-1$ goods with strictly positive values are allocated to $n$ agents, 
which inevitably implies the existence of an agent that gets nothing. This example leads to the following upper bound.

\begin{proposition}\label{pro:1/e}
For every $n,m\in \mathbb{N}$ such that $m\geq n-1$,  the $q$-quantile share is infeasible for $q > (1-\frac{1}{n})^{n-1}$. In particular, asymptotically (as $n\to \infty$) the $q$-quantile share is infeasible for $q>\frac{1}{e}$.
\end{proposition}

\begin{proof}
    Let $\vali$ be a valuation that satisfies $\vali(S_i)>0$ if $S_i\cap [n-1]\neq \emptyset$ and $\vali(S_i)=0$ otherwise.
    Every allocation $\allocs$ has at least one agent $i$ who receives none of the goods from $[n-1]$, and hence $\vali(S_i)=0$. In the random allocation, for agent $i$'s bundle $X_i$ we have 
    \begin{align*}
        \mathbb{P}[\vali(X_i)=0]=\left(1-\frac{1}{n}\right)^{n-1}.
    \end{align*}
    Therefore, for $q > (1-\frac{1}{n})^{n-1}$ the $q$-quantile of $\vali(X_i)$ is strictly positive, and hence the bundle $S_i$ with $\vali(S_i)=0$ is not $q$-fair towards agent $i$. 
\end{proof}

\begin{remark}
Later, we will consider special classes of valuations. 
In each of these classes, there is a valuation $\vali$ that satisfies $\vali(S_i)>0$ iff $S_i\cap [n-1]\neq \emptyset$. Therefore, the $1/e+O(1/n)$ upper bound of Proposition \ref{pro:1/e} applies to all these classes without any modification to the proof.
\end{remark}

Propositions \ref{pro:un} and \ref{pro:1/e} precisely determine the critical value of $q$ at which $q$-quantile shares shift from being feasible to infeasible, for the case of two agents. 

\begin{corollary}\label{cor:n2}
    For $n=2$ the $\frac{1}{2}$-quantile share is the largest feasible quantile share.
\end{corollary}

However, as $n$ becomes larger the gap between the feasibility result of Proposition \ref{pro:un} and the infeasibility result of Proposition \ref{pro:1/e} increases; the largest feasible value of $q$ is located in the interval $[\frac{1}{n},(1-\frac{1}{n})^{n-1}]\approx [\frac{1}{n},\frac{1}{e}]$. The main question that we try to address in this paper is:
\begin{center}
  \emph{Are $q$-quantile shares feasible for a constant $q>0$ that is independent of $n$ and $m$?}  
\end{center}

\section{Feasibility of Quantile Shares via Erd\H{o}s Matching Conjectures}

Our main results show that the $q$-quantile share is universally feasible for a constant $q$, under the assumption that the Erd\H{o}s Matching Conjecture is true. 
We first present the result for the case of identical valuations (Section~\ref{sec:id}). Thereafter, we extend these arguments to general (non-identical) valuations (Section~\ref{sec:gen}).

\subsection{Identical Valuations}\label{sec:id}

In this section, we restrict attention to the case where all agents have an identical (monotone) valuation function $\vali=\val\in \cV$.

Our main result uncovers a surprising connection between the feasibility of quantile shares and the well-known Erd\H{o}s Matching Conjecture. We start this section by describing the conjecture and its connection to our problem. 

Erd\H{o}s' conjecture considers the following question: \textit{what is the maximum size of a family of $k$-element subsets of an $m$-element set if it has no collection of $n$ pairwise disjoint sets?} To state it, we need the following terminology and notation. An $l$-\emph{matching} is a collection of $l$ pairwise disjoint sets. Given a family of sets $\cF$, the \emph{matching number} $\nu(\cF)$ is the maximal $l$ such that an $l$-matching from $\cF$ exists. The Erd\H{o}s Matching Conjecture gives a bound on the maximum cardinality of $\cF$ subject to the condition $\nu(\cF) < n$. Concretely, the conjecture focuses on the case where the family consists of $k$-element sets over the universe $[m]$ and states the following.

\begin{conjecture*}[Erd\H{o}s Matching Conjecture \cite{erdos1965problem}]
    For every $m,k,n\in \bN$ such that $m\geq kn$, and every $\cF \subseteq \binom{[m]}{k}$ for which $\nu(\cF) < n$, we have
    \begin{equation*}
        |\cF| \leq \max\biggl\{ \binom{m}{k}-\binom{m-n+1}{k}, \binom{kn-1}{k}\biggl\}.
    \end{equation*}
\end{conjecture*}

The expressions $\binom{m}{k}-\binom{m-n+1}{k}$ and $\binom{kn-1}{k}$ have simple interpretations in this context. One strategy for constructing a large family of sets with no $n$-matching is to enforce the property that every set includes at least one element from $[n-1]$. Then an intersection between $n$ sets must occur somewhere in these $n-1$ elements. Such a construction yields $|\cF|=\binom{m}{k}-\binom{m-n+1}{k}$. Another strategy for constructing a large family of sets with no $n$-matching is to reduce the universe from which the elements are taken. Reducing it to $[kn-1]$ is sufficient to prevent an $n$-matching. Such a construction yields $|\cF|=\binom{kn-1}{k}$. The conjecture states that for every $k,n$ and $m\geq kn$, one of these two constructions is optimal (that is, it constructs the largest possible family with no $n$-matching). 

This conjecture has received considerable attention over more than half a century. The special case $n=2$ is the well-known Erd\H{o}s-Ko-Rado theorem~\cite{erdos1961intersection}. The conjecture is trivial when $k=1$, was proved by \citet{gallai1959maximal} for the case $k=2$, and was proved much more recently in a sequence of works for the case $k=3$ (\cite{frankl2012maximum,luczak2014erdHos,frankl2017maximum}). The conjecture was established for sufficiently large $m$ (compared to $n$ and $k$) by many authors; the conditions on how large $m$ needs to be have become weaker over time, but they are still stronger than the conjectured $m \geq kn$; see e.g.,  \cite{erdos1965problem, huang2012size, peter, petand} to mention just a few.

 We will utilize the Erd\H{o}s Matching Conjecture for the case where $n$ is fixed, $k\to \infty$ and $m=n(k+1)$. In this case it can be verified that $\binom{m}{k}-\binom{m-n+1}{k} \geq \binom{kn-1}{k}$; see Lemma \ref{lem:comp} in Appendix \ref{ap:comp}. Thus the special case of the Erd\H{o}s Matching Conjecture that we need is the following.

\begin{conjecture}[Erd\H{o}s Matching Conjecture -- special case]\label{con:EMS-sp}
    For every $n$ there exists $k_0$ such that for every $k\geq k_0$, $m=(k+1)n$, and every $\cF \subseteq \binom{[m]}{k}$ for which $\nu(\cF) < n$, we have
    \begin{equation*}
        |\cF| \leq \binom{m}{k}-\binom{m-n+1}{k}.
    \end{equation*}
\end{conjecture}

To the best of our knowledge, the Erd\H{o}s Matching Conjecture remains a conjecture in this special case.\footnote{A somewhat ``close" region in which the conjecture is known to be true is $kn \leq m \leq (k+\epsilon_k)n$ for some constant $0<\epsilon_k<1$ that does not depend on $n$ \cite{frankl2017proof}.}

\paragraph{Connection to quantile shares.} To establish a connection between our problem and the Erd\H{o}s Matching Conjecture we utilize Proposition~\ref{pro:equiv} and consider $0/1$-valuations. 
The connection to the Erd\H{o}s Matching Conjecture follows from the following analogies. We set $\cF=\{\alloci\subseteq [m]:u(\alloci)=1\}$. Namely, $\cF$ is the collection of subsets in which an agent gets a value of $1$ and hence is satisfied. The notion of an $n$-matching of subsets of $[m]$ corresponds to an allocation: we cannot allocate the same good to two different agents; i.e., pairwise disjointness. With this interpretation, the Erd\H{o}s Matching Conjecture states that if no allocation yields every agent a value of $1$ (i.e., $\nu(\cF)<n$), then there are not too many subsets in which an agent has a value of $1$ (i.e., $|\cF|$ is bounded from above).

Despite this tight connection of the two problems, there is an obvious obstacle: In the allocation problem, we are allowed to allocate to agents different numbers of goods, whereas the Erd\H{o}s Matching Conjecture deals with $k$-subsets, namely corresponds to the case where all agents get the same number of goods $k$. Somewhat surprisingly, it turns out that a careful choice of the number of goods that we allocate to the agents,\footnote{The naive choice of $k=\left \lfloor \frac{m}{n} \right \rfloor$ does not provide a desired feasibility result. But $k=\left \lfloor \frac{m}{n} \right \rfloor-1$ does.} combined with the Kruskal-Katona theorem (see below) implies a feasibility result for a constant $q$. We formulate below the special case of the Kruskal-Katona Theorem that we utilize in the proof.

\begin{theorem*}[Lov\'asz's simplified formulation of the Kruskal-Katona Theorem \cite{lovasz1993comb}]
    Let $\cG_k \subseteq \binom{[m]}{k}$ be a family of $k$-subsets. For every $k' \leq k $ we define $\partial_{k'}\cG_k \subseteq \binom{[m]}{k'}$ by \footnote{$\partial_{k'}\cG_{k}$ is called the \emph{shadow} of $\cG_k$ on $\binom{[m]}{k'}$.} $$\partial_{k'}\cG_{k}=\{S'\in \binom{[m]}{k'}: \exists S\in \cG_k \text{ s.t. } S'\subseteq S\}.$$ If $|\cG_k|\geq \binom{m'}{k}$ for some $m'\leq m$ then $|\partial_{k'}\cG_{k}|\geq \binom{m'}{k'}$. 
\end{theorem*}

We are now ready to formulate and prove the result for identical valuations.

\begin{theorem}\label{theo:id}
If the Erd\H{o}s Matching Conjecture is true for the special case of Conjecture \ref{con:EMS-sp}, then for every $n,m\in \bN$  the $\frac{1}{2e}$-quantile share is feasible for any profile of identical valuations in $\cV$. 
\end{theorem}

\begin{proof}
By Proposition~\ref{pro:equiv} it suffices to prove that the $\frac{1}{2e}$-quantile share is feasible for every profile of identical valuations in $\cV_{01}$. We fix some $\epsilon > 0$ and prove the feasibility of the $(\frac{1}{2e} - \epsilon)$-quantile share; this will suffice because the critical value of feasibility is located on the discrete grid of $\frac{1}{n^m}$.

Note that feasibility of the $q$-quantile share for $m'$ implies feasibility of the $q$-quantile share for every $m''<m'$ because we can set the marginal contribution of the last $m'-m''$ goods to be identically 0. Therefore we can assume without loss of generality that $m$ is large enough (to be specified below). Moreover, we can choose $m$ to satisfy that $m/n$ is an integer.

Note that $|X_i|$ -- the number of goods that agent $i$ receives in a random allocation -- is distributed according to $Bin(m,1/n)$. For a fixed $n$, by the Central Limit Theorem we know that $\lim_{m\to \infty} \bP[|X_i| <m/n]=1/2$. We set $m$ to satisfy:
\begin{enumerate}
\item $m/n$ is an integer.
\item $m/n-1 \geq k_0$ for the $k_0$ in Conjecture \ref{con:EMS-sp}.
\item $\bP[|X_i|<m/n] \geq 1/2-\epsilon$.
\end{enumerate}

Let $u\in \cV_{01}$, let $\cF = \{S:u(S)=1\}$, and for each $k$ let $\mathcal{F}_k = \mathcal{F} \cap \binom{[m]}{k}$. Similarly let $\cG = \{S:u(S)=0\}$, and let $\cG_k = \mathcal{G} \cap \binom{[m]}{k}$. If $\mathcal{F}$ contains a matching of size $n$, then (by monotonicity) there is an allocation $\allocs$ with $u(\alloci) = 1$ for all $i\in[n]$. Thus we may assume that $\nu(\mathcal{F}) < n$, and in particular that $\nu(\mathcal{F}_k) < n$ for $k=m/n-1$. By the Erd\H{o}s Matching Conjecture (the special case of Conjecture \ref{con:EMS-sp}), we have
    \begin{equation*}
        |\cF_k| \leq \binom{m}{k}-\binom{m-n+1}{k},
    \end{equation*}
or equivalently 
\begin{equation*}
        |\cG_k| \geq \binom{m-n+1}{k}.
\end{equation*}
Let $k' \leq k$. The monotonicity of $u$ implies that $\partial_{k'}\cG_k \subseteq \cG_{k'}$, and therefore by the Kruskal-Katona Theorem we get $$|\cG_{k'}| \geq \binom{m-n+1}{k'}.$$
The fraction of $k'$-sets in which an agent has a $0$ value is bounded from below by:
\begin{align*}
    \frac{|\cG_{k'}|}{\binom{m}{k'}}\geq  
    & \frac{\binom{m-n+1}{k'}}{\binom{m}{k'}}=\frac{(m-k')\cdot (m-k'-1) \cdot ... \cdot (m-n-k'+2)}{m\cdot (m-1)\cdot ... \cdot (m-n+2)}= \\ 
    = &\left( 1-\frac{k'}{m} \right) \cdot \left( 1-\frac{k'}{m-1} \right) \cdot ... \cdot \left( 1-\frac{k'}{m-n+2} \right) \geq \left( 1-\frac{k'}{m-n+2} \right)^{n-1} \geq \\ 
    \geq &\left( 1-\frac{k}{m-n+2} \right)^{n-1} \geq \left( 1-\frac{k}{m-n} \right)^{n-1} = \left( 1-\frac{1}{n} \right)^{n-1} \geq \frac{1}{e}.
\end{align*}
In a random allocation, the probability that an agent will have a $0$ value is at least $(\frac{1}{2}-\epsilon) \frac{1}{e}$. Indeed, with probability at least $(\frac{1}{2}-\epsilon)$ the random bundle $X_i$ will satisfy $|X_i|<\frac{m}{n}$, i.e., $|X_i|\leq k$. Conditional on $|X_i|=k'\leq k$, the probability of having $0$ value is at least $\frac{1}{e}$ (because the conditional distribution is uniform over $\binom{[m]}{k'}$). Therefore, the $(\frac{1}{2e}-\epsilon)$-quantile of $u(X_i)$ is located at $0$ and agents are satisfied even if they get a value of $0$. 
\end{proof}

\subsection{General Valuations}\label{sec:gen}

To apply the techniques of Section \ref{sec:id} to general valuations (not necessarily identical) a stronger version of the conjecture is needed. Instead of having a single family $\cF$ (which reflects the valuation of an agent), we have $n$ possibly different families $\cF^1,...,\cF^n$, one for each agent. Interestingly, such a variant of the Erd\H{o}s Matching Conjecture has been studied in the literature; see \cite{huang2012size, aharoni1size, gao2022on, lu2023a, kupavskii2023rainbow}.

Given $\mathcal{F}^1,\ldots,\mathcal{F}^n \subseteq \binom{[m]}{k}$, a \emph{rainbow matching} in $(\mathcal{F}^1,\ldots,\mathcal{F}^n)$ is a collection of pairwise disjoint sets $S_1,\ldots,S_n$, where $S_i \in \mathcal{F}_i$ for each $i \in [n]$. The collection of families is \emph{cross-dependent} if it has no rainbow matching.

\begin{conjecture*}[Rainbow Erd\H{o}s Matching Conjecture \cite{huang2012size,aharoni1size}]
    For every $m,k,n\in \bN$ such that $m\geq kn$, and every cross-dependent collection of families $\mathcal{F}^1,\ldots,\mathcal{F}^n \subseteq \binom{[m]}{k}$, we have
    \begin{equation*}
        \min_{i\in[n]}|\mathcal{F}^i| \leq \max\biggl\{\binom{m}{k}-\binom{m-n+1}{k},\binom{kn-1}{k}\biggl\}.
    \end{equation*}
\end{conjecture*}

The Rainbow Erd\H{o}s Matching Conjecture generalizes the Erd\H{o}s Matching Conjecture because one can set $\cF^i=\cF$ for all $i\in [n]$ which gives precisely the Erd\H{o}s Matching Conjecture. Similarly to Section \ref{sec:id}, we will need the validity of the conjecture for a special case.

\begin{conjecture}[Rainbow Erd\H{o}s Matching Conjecture - special case]\label{con:REMS-sp}
    For every $n$ there exists $k_0$ such that for every $k\geq k_0$, $m=(k+1)n$, and every cross-dependent collection of families $\mathcal{F}^1,\ldots,\mathcal{F}^n \subseteq \binom{[m]}{k}$, we have
    \begin{equation*}
        \min_{i\in[n]}|\mathcal{F}^i| \leq \binom{m}{k}-\binom{m-n+1}{k}.
    \end{equation*}
\end{conjecture}

Analogously to the case of identical valuations, we have the following result for general monotone valuations.

\begin{theorem}\label{theo:gen}
If the Rainbow Erd\H{o}s Matching Conjecture is true for the special case of Conjecture~\ref{con:REMS-sp}, then for every $n,m\in \bN$  the $\frac{1}{2e}$-quantile share is universally feasible.
\end{theorem}

\begin{proof}
As in the proof of Theorem~\ref{theo:id}, it suffices to consider a profile of $0/1$-valuations $u_1,\ldots,u_n \in \cV_{01}$, and prove the feasibility of the $(\frac{1}{2e}-\epsilon)$-quantile share (for arbitrary $\epsilon>0)$. Moreover, we may assume that the $(\frac{1}{2e}-\epsilon)$-quantile of every $u_i$ is equal to $1$. Indeed, any $u_i$ whose quantile is $0$ places no constraints on the allocation, so we may replace such $u_i$ by an arbitrary $u'_i$ whose quantile is $1$. Furthermore, we set $m$ to satisfy conditions 1--3 as in the previous proof.

Assume for the sake of contradiction that no allocation ensures $u_i(\alloci)=1$ for all $i\in [n]$. We define $\cF^i = \{\alloci:u_i(\alloci)=1\}$, and let $\mathcal{F}^i_k = \mathcal{F}^i \cap \binom{[m]}{k}$. Thus, for $k=m/n-1$, the collection $\cF^1_k,\ldots,\cF^n_k$ is cross-dependent. By the Rainbow Erd\H{o}s Matching Conjecture (the special case of Conjecture \ref{con:REMS-sp}), we have
    \begin{equation*}
        |\cF^i_k| \leq \binom{m}{k}-\binom{m-n+1}{k}
    \end{equation*}
    for some $i\in [n]$. We repeat the same arguments as in the proof of Theorem \ref{theo:id} to deduce that in a random allocation this particular agent $i$ must have a probability of at least $(\frac{1}{2}-\epsilon) \frac{1}{e}$ to have a 0 value. This contradicts the fact that the $(\frac{1}{2e}-\epsilon)$-quantile of $u_i$ is located at $1$. 
\end{proof}

\section{Unconditional Feasibility Results}
\label{sec:unconditional}

Theorems \ref{theo:id} and \ref{theo:gen} provide quite surprising and reasonably tight bounds on the critical value of feasibility for quantile shares. In particular, asymptotically (as $n \to \infty$), the critical threshold between feasibility and infeasibility is conjectured to reside in $[\frac{1}{2e},\frac{1}{e}]$ (we recall the bound of Proposition \ref{pro:1/e}). An obvious shortcoming of these results is the fact that they rely on conjectures (well-known conjectures, but yet conjectures). In this section, we present some unconditional positive results for special classes of valuations. For any valuation function $\val:2^{[m]} \to \bR_+$ we denote by $\val(j \mid S)$ the \emph{marginal value} of $j \in [m]$ given the set $S \subseteq [m]$, that is $\val(j\mid S) = \val(S \cup \{j\}) - \val(S)$.

\subsection{Additive Valuations}
The class of additive valuations is the most well-studied class of valuations in the context of fairness. We denote by $w(i,j)$ the value agent $i$ has for good $j$ (where $w(i,j)\geq 0$ for all $i,j$).
\begin{definition}
    The valuation function $\vali$ is \emph{additive} if $\vali(S) = \sum_{j\in S} w(i,j)$ for all $S\subseteq[m]$.
\end{definition}
We prove the following feasibility result for constant values of $q$.

\begin{proposition}\label{pro:add}
    For every $n,m\in \bN$ the $0.14(1-\frac{1}{n})^n$-quantile share is feasible for the class of additive valuations. 
    In particular, asymptotically (as $n\to \infty$) the $q$-quantile share is feasible for every $q<\frac{0.14}{e}$. 
\end{proposition}

In comparison with Theorem \ref{theo:gen}, this Proposition provides a worse bound ($\frac{0.14}{e}$ versus $\frac{0.5}{e}$) and is applicable to additive valuations only. However, it does not rely on any conjectures.

The proof of Proposition \ref{pro:add} relies on \emph{deviation of sums inequalities}. These inequalities bound the probability that the sum of independent random variables will exceed its mean. Several such inequalities have been suggested in the literature \cite{samuels1966on,feige2006sums,garnett2020small}.\footnote{Interestingly, some connections between the deviation of sums inequalities and the Erd\H{o}s Matching Conjecture have been established \cite{luczak2017on,petand}.} For our purposes, the special case of Bernoulli random variables will play a role. For this special case, the following inequality has been established.
\begin{lemma*}[\citet{arieli2020speed}]
        For every $p\in [0,1]$, $m\in \bN$ and $w_1,\ldots,w_m \geq 0$, if $b_1,\ldots,b_m$ are i.i.d. Bernoulli$(p)$ random variables, then
        \begin{equation*}
            \bP\left[\sum_{j=1}^{m}w_jb_j \geq \left(\sum_{j=1}^{m}w_j\right)p\right] \geq 0.14 p.
        \end{equation*}
\end{lemma*}
This Lemma follows from the inequality of \citet{feige2006sums} and its subsequent improvement by \citet{garnett2020small}. By flipping the roles of $0$ and $1$ in the Bernoulli random variables we get the following equivalent formulation. 
\begin{lemma} \label{lem:ber}
        For every $p\in [0,1]$, $m\in \bN$ and $w_1,\ldots,w_m \geq 0$, if $b_1,\ldots,b_m$ are i.i.d. Bernoulli$(p)$ random variables, then
        \begin{equation*}
            \bP\left[\sum_{j=1}^{m}w_jb_j \leq \left(\sum_{j=1}^{m}w_j\right)p\right] \geq 0.14 (1-p).
        \end{equation*}
\end{lemma}

We now turn to the proof of Proposition \ref{pro:add}.
\begin{proof}[Proof of Proposition \ref{pro:add}]
    We will show that the round-robin algorithm terminates with an allocation in which every agent is $0.14(1-\frac{1}{n})^n$-satisfied. 
    The round-robin algorithm has $m$ steps. In every step $t=dn+i \in [m]$ the algorithm allocates to agent $i$ her most preferable good from the remaining $m-t+1$ goods, breaking ties in favor of the lowest-indexed good. 

    We first show that agent 1 ends up being $0.14(1-\frac{1}{n})$-satisfied in the round-robin algorithm. 
    For simplicity of notation, we let $w_j=w(1,j)$ be the value of agent 1 for good $j$.
    Assume without loss of generality that $w_1\geq \cdots \geq w_m$; namely that agent 1's preferences over goods are in decreasing order. We denote by $a_1,...,a_k$ the goods that were allocated to agent $1$. 
    We denote by $W^1_{RR}=w_{a_1}+...+w_{a_k}$ the value of agent 1 in the round-robin algorithm.
    Note that $a_1\leq 1$, $a_{2}\leq n+1$,..., and $a_k\leq (k-1)n+1$, because in step $t=dn+1$, in the worst case, the goods $[dn]$ were already allocated. Therefore,
    \begin{align*}
        W^1_{RR} &\geq w_1 + w_{n+1} + \cdots + w_{(k-1)n+1} \geq \\
        &\geq \frac{1}{n}[w_1+...+w_n]+\frac{1}{n}[w_{n+1}+...+w_{2n}]+...+\frac{1}{n}[w_{(k-1)n+1}+...+w_m] = \frac{1}{n}\sum_{j\in[m]} w_j.
    \end{align*}

    For additive valuations, the value of agent 1 in a random allocation can be written as $\val_1(X_1)=\sum_{j\in [m]} w_j b_j$ where $b_1,...,b_m$ are i.i.d. Bernoulli$(\frac{1}{n})$ random variables. By Lemma \ref{lem:ber} we get that with probability at least $0.14(1-\frac{1}{n})$ her realized value will be (weakly) below the expectation $\frac{1}{n}\sum_{j\in[m]} w_j$ and hence (weakly) below what she actually gets in the round-robin algorithm: $W^1_{RR}$. Namely agent~1 is $0.14(1-\frac{1}{n})$-satisfied.
    
    Now we turn to prove that every agent $i=2,...,n$ is $0.14(1-\frac{1}{n})^i$-satisfied. We observe that after $i-1$ steps of the round-robin algorithm agent $i$ plays the role of agent $1$ with one difference: a set of $i-1$ goods, which we denote by $A\subseteq [m]$ has already been eliminated from the pool of goods. 
    We denote by $E$ the event that agent $i$ does not get any good from $A$ in a random allocation. Note that $\bP [E]=(1-\frac{1}{n})^{i-1}$. 

    We repeat the above arguments for agent $i$ instead of agent $1$ when we condition the random bundle $X_i$ on the event $E$. We denote by $W^i_{RR}$ the value of agent $i$ in the round-robin algorithm.  By the arguments above we get $\bP[\vali(X_i)\leq W^i_{RR}|E]\geq 0.14(1-\frac{1}{n})$. Therefore, $$\bP[\vali(X_i)\leq W^i_{RR}]\geq \bP[E]\cdot \bP[\vali(X_i)\leq W^i_{RR}|E]\geq \left(1-\frac{1}{n}\right)^{i-1} \cdot 0.14  \left(1-\frac{1}{n}\right)=0.14 \left(1-\frac{1}{n}\right)^i.$$ 
    Namely, agent $i$ is $0.14(1-\frac{1}{n})^i$-satisfied. Hence every agent is $0.14(1-\frac{1}{n})^n$-satisfied.
\end{proof}

\subsection{Unit-Demand Valuations}
In the class of unit-demand valuations every good $j\in [m]$ has a value of $w(i,j)\geq 0$ for agent $i$.
\begin{definition}
    The valuation function $\vali$ is \emph{unit-demand} if $\vali(S) = \max_{j\in S} w(i,j)$ for all $S\subseteq[m]$.
\end{definition}
We prove the following tight feasibility result for $q=(1-\frac{1}{n})^{n-1}$. The tightness follows from Proposition \ref{pro:1/e}.

\begin{proposition}\label{pro:ud}
    For every $n,m\in \bN$ the $(1-\frac{1}{n})^{n-1}$-quantile share is feasible for the class of unit-demand valuations. 
    In particular, the $\frac{1}{e}$-quantile share is feasible for this class. 
\end{proposition}

\begin{proof}
The proof is similar to that of Proposition \ref{pro:add} but is, in fact, simpler. 
We consider the round-robin algorithm, and we observe that agent 1 is $1$-satisfied, because she is allowed to pick her most favorable good. By the arguments in the proof of Proposition \ref{pro:add} we deduce that agent $i$ is $(1-\frac{1}{n})^{i-1}$-satisfied. Thus every agent is $(1-\frac{1}{n})^{n-1}$-satisfied. 
\end{proof}

\subsection{Matroid-Rank Valuations}
A monotone valuation function $v$ is submodular if the marginal contribution of a good decreases as the set increases, i.e., $\val(j\mid S') \leq \val(j\mid S)$ for $S\subseteq S'$. Unfortunately, an unconditional feasibility proof of $q$-quantile shares for constant $q$ remains elusive for submodular valuations. 

However, there is an important subclass of submodular valuations for which we can prove that $\frac{1}{e}$ is the critical threshold for feasibility for large values of $n$, without relying on conjectures. These are the \emph{matroid-rank} valuations, namely those valuation functions $v: 2^{[m]} \to \bN_0$ (where $\bN_0=\bN \cup \{0\}$), for which there exists a matroid $M$ on $[m]$ so that $\val(S)$ is the rank of $S$ in $M$.
\begin{definition}
    The valuation function $\vali$ is \emph{matroid-rank} if $\vali$ is the rank function of some matroid $M = ([m],\mathcal{I})$ over the ground set $[m]$. The rank function assigns to each set $S\subseteq[m]$ the cardinality of a largest independent subset of $S$, i.e., $\vali(S) = \max_{I\in\mathcal{I}, I \subseteq S} |I|$.
\end{definition}
It is known that these are precisely the submodular valuations $v$ which satisfy $v(\emptyset)=0$ and $\val(j\mid S)\in \{0,1\}$ for every $S$ and every $j$. 
The literature has identified several kinds of resource-allocation settings where matroid-rank valuations arise naturally: see e.g., \cite{benabbou2020finding, babaioff2021fair}. Typically, those are contexts in which the agents' values are determined by solving (suitably structured) combinatorial optimization problems.

\begin{proposition}\label{pro:mrf}
   For every $n,m\in \bN$ and $q=\max \{\frac{1}{e}-\frac{1}{2\sqrt{n(n-1)}},\frac{1}{n}\}$ the $q$-quantile share is feasible for the class of matroid-rank valuations. 
   In particular, asymptotically (as $n\to \infty$) the $q$-quantile share is feasible for every $q<\frac{1}{e}$.
\end{proposition}

Note that $\frac{1}{e}-\frac{1}{2\sqrt{n(n-1)}}>\frac{1}{n}$ for $n\geq 5$, so the $\frac{1}{n}$ term in the maximum is relevant only for $n=2,3,4$. The proof below utilizes the underlying matroid structure of the valuations and the powerful Edmonds' Matroid Intersection Theorem.

\begin{theorem*}[Edmonds' Matroid Intersection Theorem~\cite{edmonds1970submodular}]
Let $M_1,M_2$ be two matroids on the same ground set $E$, with respective families of independent sets $\mathcal{I}_1, \mathcal{I}_2$ and rank functions $\rho_1,\rho_2$. We have
$$ \max_{I \in \mathcal{I}_1 \cap \mathcal{I}_2} |I| = \min_{A \subseteq E} [\rho_1(A) + \rho_2(E \setminus A)].$$
\end{theorem*}

\begin{proof}[Proof of Proposition~\ref{pro:mrf}]
The feasibility of the $\frac{1}{n}$-quantile share has been proved in Proposition \ref{pro:un}. It remains to prove the feasibility of the $(\frac{1}{e}-\frac{1}{2\sqrt{n(n-1)}})$-quantile share.

The feasibility of the maximin share for matroid-rank valuations was shown by \citet{barman2020existence}. It is sufficient to prove that for matroid-rank valuations the maximin share has a quantile of at least $\frac{1}{e}-\frac{1}{2\sqrt{n(n-1)}}$. We fix an agent $i$ with matroid-rank valuation $v$ and we omit the agent's index notation hereafter for clarity of notations.

We denote by $M$ the matroid over the ground set $[m]$ that represents $v$. Namely, $v(S)$ is the maximum size of an independent set of $M$ that is contained in $S$. 
We denote by $k$ the maximin share of $v$. Namely, $k$ is the maximal value $k'$ for which there exist $n$ disjoint independent (in $M$) sets $S_1,...,S_n$ with $|S_j|=k'$.

Let $\uv(S)=\min \{v(S),k+1\}$ which is also a matroid-rank valuation, and let $\uM$ be the corresponding matroid.
We define two matroids over the ground set $[n]\times [m]$.

\begin{itemize}
    \item $M^{\oplus}$ is the direct sum of $n$ copies of $\uM$. Namely, its independent sets are those $S\subseteq [n]\times [m]$ such that for every $i\in [n]$ the set $\{j\in [m]:(i,j)\in S\}$ is independent in $\uM$. The corresponding rank function is denoted by $\rho_{M^{\oplus}}:2^{[n]\times [m]} \to \mathbb{N}_0$.
    \item $N$ is the partition matroid with respect to the blocks $[n]\times \{j\}$ for $j\in [m]$. Namely, its independent sets are those $S\subseteq [n]\times [m]$ such that for every $j\in [m]$ we have $|\{i:(i,j)\in S\}|\leq 1$. The corresponding rank function is denoted by $\rho_{N}:2^{[n]\times [m]} \to \mathbb{N}_0$. 
\end{itemize}

Note that a common independent set of $M^{\oplus}$ and $N$ corresponds to a collection of $n$ disjoint independent  sets of $\uM$. Since $k$ is the maximin share we know that there is no common independent set of $M^{\oplus}$ and $N$ of size $(k+1)n$. Now Edmonds' Matroid Intersection Theorem implies the existence of a subset $A\subseteq [n]\times [m]$ such that 
\begin{align}\label{eq:edmond}
   \rho_{M^{\oplus}}(A)+\rho_N(([n]\times [m]) \setminus A)<(k+1)n. 
\end{align}
We denote $A_i=\{j\in [m]:(i,j)\in A\}$. Equation \eqref{eq:edmond} can be equivalently written as
\begin{align}\label{eq:ee}
    \sum_{i\in [n]} \uv(A_i)+|\bigcup_{i\in [n]} ([m] \setminus A_i)|<(k+1)n.
\end{align}
Replacing each $A_i$ by $A_0=\cap_{i\in [n]} A_i$ weakly decreases the left-hand side of Equation \eqref{eq:ee} because the second term remains unchanged while the first term weakly decreases. Therefore we get
\begin{align}\label{eq:ee2}
    n\cdot \uv(A_0)+m-|A_0|<(k+1)n.
\end{align}
This implies that $\uv(A_0)\leq k$. Hence, writing $t=k+1-\uv(A_0)$, we have $t\geq 1$. With this notation Equation \eqref{eq:ee2} is equivalent to 
\begin{align}\label{eq:tn}
   m-|A_0|\leq tn-1.
\end{align}
For every bundle $X\subseteq [m]$ we argue that the condition $|X\setminus A_0|\leq t-1$ is sufficient to ensure $v(X)\leq k$, namely that the agent's value in the bundle $X$ is weakly below the maximin share. Indeed,
\begin{align*}
    \uv(X)\leq \uv(X\cap A_0)+|X\setminus A_0|\leq \uv(A_0)+|X\setminus A_0|\leq \uv(A_0)+t-1 =k \Rightarrow v(X)\leq k.
\end{align*}
For a random bundle $X$ that gets every good with probability $1/n$, the distribution of $|X\setminus A_0|$ is binomial with $m-|A_0|$ trials and probability of success $1/n$. This distribution is stochastically dominated by a binomial distribution $Y$ with $tn-1$ trials and probability of success $1/n$ (by Equation~\eqref{eq:tn}).
Therefore, it is sufficient to prove that for every $n\geq 2, t\geq 1$, and $Y\sim Bin(tn-1,\frac{1}{n})$ we have $\mathbb{P}[Y<t]\geq \frac{1}{e}-\frac{1}{2\sqrt{n(n-1)}}$.

Let $Z\sim Poisson(\frac{tn-1}{n})$. \citet{romanowska1977note} bounded the total variation distance between any binomial distribution with success probability $p$ and its approximating Poisson distribution  by $\frac{p}{\sqrt{1-p}}$. In our case $p=\frac{1}{n}$, so this bound becomes $\frac{1}{\sqrt{n(n-1)}}$. It follows that for every subset $R$ of $\bN_0$ we have $|\mathbb{P}[Y \in R] - \mathbb{P}[Z \in R]| \leq \frac{1}{2\sqrt{n(n-1)}}$. Therefore we can deduce that
\begin{align*}
        \mathbb{P}[Y<t]=\mathbb{P}[Y\leq \frac{tn-1}{n}] \geq \mathbb{P}[Z\leq \frac{tn-1}{n}] -\frac{1}{2\sqrt{n(n-1)}} \geq \frac{1}{e}-\frac{1}{2\sqrt{n(n-1)}}.
\end{align*}
The last inequality follows from Teicher \cite{teicher1955inequality} who proved that the realization of any Poisson distribution is weakly below its expectation with probability greater than $1/e$.
\end{proof}

\begin{remark}
For $n \geq 2$ let us denote $$q_n = \inf_{t \in \bN} \mathbb{P}[Y_t < t], \,
\text{ where } Y_t \sim Bin(tn-1,\frac{1}{n}).$$ We showed in the proof above that for any given $n$, the $q_n$-quantile share is feasible for the class of matroid-rank valuations. 
In this form, the result is actually tight: take $m=tn-1$, and let each agent's valuation be represented by the uniform matroid of rank $t$ over $[m]$. We conjecture that in fact $q_n = (1-\frac{1}{n})^{n-1}$, i.e., for any given $n$ the infimum is attained at $t=1$. If true, this would show that $(1-\frac{1}{n})^{n-1}$ is the critical value for feasibility of quantile shares in the class of matroid-rank valuations, 
for any given $n$. 
While we are unable to prove this conjecture exactly, in Proposition~\ref{pro:mrf} we estimate $q_n$ up to an error term which vanishes as $n \to \infty$.
\end{remark}

\subsection{Supermodular Valuations}
A valuation function is supermodular if the marginal contribution of a good increases as the set increases.
\begin{definition}
    The valuation function $\vali$ is \emph{supermodular} if $\vali(j\mid S') \geq \vali(j\mid S)$ for $S\subseteq S' \subseteq [m] \setminus \{j\}$.
\end{definition}
The class of supermodular monotone valuations is as general as the class of all monotone valuations in the context of feasibility of quantile shares.

\begin{proposition}\label{lem:sup}
 For every $q\in [0,1]$, if the $q$-quantile share is feasible for the class of supermodular monotone valuations, 
 then the $q$-quantile share is universally feasible.
\end{proposition}

\begin{proof}
    Given $q\in [0,1]$, $i\in [n]$ and a monotone valuation $\vali:2^{[m]} \to \bR_+$, we construct a supermodular monotone valuation $u_i:2^{[m]} \to \mathbb{R}_+$ as follows.

    The valuation $\vali$ induces a weak total order $\preceq_{v_i}$ over the set of bundles $2^{[m]}$. We break ties in an arbitrary monotonic manner to derive a \emph{strict} total order $\prec_i$ over the set of bundles $2^{[m]}$. It has been proved by Chambers and Echenique \cite{chambers2009supermodularity} that there exists a supermodular valuation $u_i$ that has the same strict total order $\prec_i$ over the set of bundles.

    By the assumption that the $q$-quantile share is feasible for $u_1,\ldots,u_n$, we get that there exists an allocation in which $u_i(S_i)$ is located weakly above the $q$-quantile of $u_i(X_i)$. Note that the same allocation places $v_i(S_i)$ weakly above the $q$-quantile of $v_i(X_i)$, because for every realization $T_i$ of $X_i$ we have $v_i(S_i)<v_i(T_i) \Rightarrow u_i(S_i)<u_i(T_i)$. Therefore the $q$-quantile share is universally feasible.
\end{proof} 

\section{Discussion}
\subsection{The Gap Between the Constants}
Assuming the Erd\H{o}s Matching Conjectures are true, we have shown that the largest value of $q$ for which the $q$-quantile share is universally feasible lies in the interval $[\frac{1}{2e},\frac{1}{e}]$. It remains an open problem to close this $\frac{1}{2e}$ gap between the two bounds. 

We discuss directions to improve the $\frac{1}{2e}$ bound (Theorems \ref{theo:id} and \ref{theo:gen}). All the techniques for proving feasibility results in this paper focus on allocations with almost equal-size bundles for all agents. In particular Theorems \ref{theo:id} and \ref{theo:gen} allocate to every agent $k=\frac{m}{n}-1$ goods, and do not specify how to allocate the remaining $n$ goods. But in any case, no agent will have more than $k+n$ goods. The round-robin algorithm in Propositions \ref{pro:add} and \ref{pro:ud} allocates to every agent $\lfloor \frac{m}{n} \rfloor$ or $\lceil \frac{m}{n} \rceil$ goods. The following example demonstrates that in order to improve the $\frac{1}{2e}$ bound we \emph{must} exploit the possibility of allocating goods unequally. In other words, it shows that the $\frac{1}{2e}$ bound is tight if every agent must get (approximately) the same number of goods.

\begin{example}
Let $\delta > 0$ be arbitrarily small. We will construct instances of the allocation problem with $n$ agents and $m$ goods (where $1 << n << m$) satisfying: for every allocation $(S_1,\ldots,S_n)$ in which $|S_i| \leq \frac{m}{n} + o(\sqrt{\frac{m}{n}})$ for all $i \in [n]$, there exists an agent who is not $(\frac{1}{2e} + \delta)$-satisfied.

First, we choose $n$ large enough so that $(1-\frac{1}{n})^{n-1} \leq \frac{1}{e} + \delta$. Next, for any such $n$, we choose $m$ large enough so that for $Y \sim Bin(m-n+1, \frac{1}{n})$ we will have, by the Central Limit Theorem, $\mathbb{P}[Y \leq \frac{m}{n} + o(\sqrt{\frac{m}{n}})] \leq \frac{1}{2} + \delta$. Finally, for any such $n$ and $m$, we choose $\epsilon > 0$ small enough so that $\epsilon(\frac{m}{n} + o(\sqrt{\frac{m}{n}})) < 1$.

For these choices of $n,m$ and $\epsilon$, consider identical additive valuations for all agents, in which the value of every good $j \in [n-1]$ is $w_j=1$, and the value of every good $j \in \{n, n+1,\ldots,m\}$ is $w_j=\epsilon$; we call the former $1$-goods and the latter $\epsilon$-goods.

Let $(S_1,\ldots,S_n)$ be an allocation in which $|S_i| \leq \frac{m}{n} + o(\sqrt{\frac{m}{n}})$ for all $i \in [n]$. Let $i$ be an agent who gets no $1$-good, and therefore has $v_i(S_i) \leq \epsilon(\frac{m}{n} + o(\sqrt{\frac{m}{n}})) < 1$. In a random allocation we have $v_i(X_i) \leq \epsilon(\frac{m}{n} + o(\sqrt{\frac{m}{n}}))$ exactly when the following two indepndent events happen: agent $i$ gets no $1$-good, and at most $\frac{m}{n} + o(\sqrt{\frac{m}{n}})$ $\epsilon$-goods. By our choices above, the probability of these two events happening is at most $(\frac{1}{e} + \delta)(\frac{1}{2} + \delta) < \frac{1}{2e} + \delta$ (here we assume, w.l.o.g., that $\delta < \frac{1}{2} - \frac{1}{e}$). This shows that $S_i$ is not $(\frac{1}{2e} + \delta)$-fair towards agent $i$, as claimed.

Note that in this example the $\frac{1}{e}$-quantile share is at most $\epsilon(m-n+1)$, the total value of the $\epsilon$-goods. Hence, when $\epsilon(m-n+1) \leq 1$, we can give all the $\epsilon$-goods to one agent and one $1$-good to every other agent, so that everyone will be $\frac{1}{e}$-satisfied. This allocation uses bundles whose sizes significantly differ.
\end{example}

We tend to conjecture that $\frac{1}{e}$ for $n \to \infty$ (and more ambitiously $(1-\frac{1}{n})^{n-1}$ for any given $n$) is the correct critical threshold for the feasibility of quantile shares. For special classes of valuations such as unit-demand and matroid-rank functions this was proved in Propositions \ref{pro:ud} and \ref{pro:mrf}. Another evidence is that for $n=2$ the critical value is $(1-\frac{1}{n})^{n-1}=\frac{1}{2}$ (see Corollary \ref{cor:n2}). Below we show that also for $n=3$ and low values of $m$ the critical value is $(1-\frac{1}{n})^{n-1}=\frac{4}{9}$.

\begin{proposition}\label{lem:n3m6}
    For $n=3$ and for $2\leq m \leq 6$ the $\frac{4}{9}$-quantile share is the largest universally feasible quantile share. 
\end{proposition}
The infeasibility of any larger quantile share is stated in Proposition \ref{pro:1/e}. The feasibility of the $\frac{4}{9}$-quantile share is proved by utilizing the fact that it suffices to consider $0/1$-valuations, and a carefully chosen case analysis.\footnote{Such techniques seem to be inapplicable for large values of $n$ and $m$.} The detailed proof is relegated to Appendix \ref{ap:m6}.

\subsubsection{Results from Computer Simulation} \label{subsec:simulation}

In addition to the above results, we conduct a computer simulation towards verifying the above conjecture. The simulation is based on exhaustive search, and certifies that the critical value is exactly $(1-\frac{1}{n})^{n-1}$ for specified input values of $n$ and $m$; i.e., the example described in Proposition~\ref{pro:1/e} is the worst case.

As we showed previously, it is sufficient to consider monotone $0/1$-valuations. Note that in the example in Proposition~\ref{pro:1/e}, the agent under consideration has value 1 in exactly $n^{m-n+1}\cdot[n^{n-1}-(n-1)^{n-1}]$ allocations, which is the number of allocations in which that agent receives an item in the set $[n-1]$. Consequently, our goal is to prove that if every agent has value 1 in at least $n^{m-n+1}\cdot[n^{n-1}-(n-1)^{n-1}]+1$ allocations (out of the total $n^m$ allocations), then there is an allocation in which every agent has value 1. Equivalently, there is no instance in which every agent has value 0 in at most $n^{m-n+1}\cdot[(n-1)^{n-1}]-1$ allocations, but in which none of the allocations satisfies every agent.

For specified values of $n$ and $m$, we prove this statement by solving the following integer program using Gurobi, a commercially-available IP solver. The integer program has a 0/1 variable $x_{(i,S)}$ for every agent $i$ and every subset $S$ of the goods, indicating the value that agent $i$ has for the set $S$ (therefore the collection of variables $(x_{(i,S)}:i\in[n],S\in2^{[m]})$ together specify the complete profile of agent valuations). We then add the following sets of constraints:
\begin{itemize}
\item \textit{monotonicity constraints}, which enforce monotonicity on every agent's valuation; i.e., for each agent $i$ and nonempty set $S$ we add the constraints $x_{(i,S)} \geq x_{(i,S')}$ for all $S' \subset S: |S'| = |S|-1$;
\item a \textit{threshold constraint} for every agent, which enforces that the number of $0$-valued allocations for that agent is at most $n^{m-n+1}\cdot[(n-1)^{n-1}]-1$; and
\item an \textit{allocation constraint} for every allocation, which ensures that some agent is unhappy, i.e. receives a set of value 0, in that allocation.
\end{itemize}

The above integer program is computationally tractable for $n=3$ and $m\leq 9$, and for $n=4,5$ and $m\leq8$. For all of these values, Gurobi reported the infeasibility of the above program, proving that the critical value is indeed $(1-\frac{1}{n})^{n-1}$. As a sanity check, we modify the threshold constraints to increase the threshold by one, that is, we allow for the number of 0-valued allocations for each agent to be at most $n^{m-n+1}\cdot[(n-1)^{n-1}]$. In each of the above cases, the solver discovered a feasible solution under the new threshold constraints. The above experiments lead to the following proposition.

\begin{proposition}\label{lem:n345m89}
    For $n=3$ and $m \leq 9$, and for $n=4,5$ and $m\leq8$, the $(1-\frac{1}{n})^{n-1}$-quantile share is the largest universally feasible quantile share. 
\end{proposition}

\subsection{Comparison of Quantile Shares with Other Notions of Shares}\label{subsec:comparison}
As mentioned above, the two most extensively studied notions of shares are the maximin share and the proportional share.
A natural question to study when comparing these notions of shares is the following: Assume that a bundle is fair towards agent $i$ with respect to the maximin share. Does this imply that it is also fair with respect to the notion of quantile shares studied here? Or equivalently: \textit{Is there a good lower bound on the quantile of the maximin share?}\footnote{No upper bound on the quantile of the maximin share can be bounded away from 1 as $n\to \infty$. For example, if there is a single good to allocate then the maximin share is 0 and its quantile is $1-\frac{1}{n}$. The same example demonstrates that the quantile of the proportional share might be as high as $1-\frac{1}{n}$.} Similarly, we can ask the same question for the proportional share.

\subsubsection{Maximin Share}

Interestingly, all the feasibility results in the paper (conditional or unconditional on conjectures) in the case of identical valuations can be equivalently viewed as lower bounds on the quantile of the maximin share. This is implied by the following general observation. We denote by $\tau_{MM}(v_i,n)$ the maximin share.

\begin{proposition}\label{lem:mm}
    Let $\mathcal{U}\subseteq \mathcal{V}$ be a class of valuations. The $q$-quantile share is feasible for every profile of $n$ identical valuations in $\mathcal{U}$ if and only if the quantile of $\tau_{MM}(u,n)$ is at least $q$ for all $u\in \mathcal{U}$.
\end{proposition}

\begin{proof}
    If the $q$-quantile share is feasible then for every $u\in \mathcal{U}$, an allocation that is $q$-fair towards all the agents (having valuation $u$) witnesses that the quantile of the maximin share is at least $q$. Conversely, if the quantile of $\tau_{MM}(u,n)$ is at least $q$, then a partition of $[m]$ into $n$ bundles attaining the maximin value can be viewed as an allocation that is $q$-fair towards all the agents (having valuation $u$).
\end{proof}

Proposition \ref{lem:mm} implies in particular that the quantile of the maximin share is always at least $\frac{1}{2e}$ assuming the Erd\H{o}s Matching Conjecture (Theorem \ref{theo:id}). Moreover, an unconditional asymptotic lower bound of $\frac{0.14}{e}$ (respectively, $\frac{1}{e}$) is valid for the additive (respectively, unit-demand and matroid-rank) class of valuations as a corollary to Proposition~\ref{pro:add} (respectively, Propositions~\ref{pro:ud} and \ref{pro:mrf}).

\subsubsection{Fractions of the Maximin Share}
As mentioned above, an active research direction has been to derive feasibility results for fractions of the maximin share in cases where the maximin share is infeasible. Unlike quantile shares, this direction is hopeless for general valuations (see Appendix \ref{ap:max-min}). 

In view of the feasibility of constant quantile shares, as opposed to the infeasibility of high enough fractions of the maximin share, one might hypothesize that quantile shares are less demanding fairness criteria than fractions of the maximin share.
The following example demonstrates that the above hypothesis is wrong in general; namely there are (simple) instances in which quantile shares are more demanding fairness notions than fractions of the maximin share.

\begin{example}
    Consider the case in which many identical goods with value 1 each (additively) are allocated. For every $\epsilon>0$ the $(1-\epsilon)$-maximin share is located at $(1-\epsilon) \lfloor \frac{m}{n} \rfloor$. On the other hand, by the Central Limit Theorem, for every $q>0$ the $q$-quantile share is located at $\frac{m}{n}-\Theta_q(\sqrt{\frac{m}{n}})$.

    Namely, for every $n \in \bN$ and $\epsilon,q>0$, for sufficiently large $m$, the $q$-quantile share notion is more demanding here than the $(1-\epsilon)$-maximin share.
\end{example}

 The above example indicates that the positive results for quantile shares are derived for general valuations not because quantiles are less demanding, but because they measure fairness in different units which are arguably more suitable for general valuations.

\subsubsection{Proportional Share}

The proportional share makes sense mainly for additive valuations. Its exclusive focus on the full bundle $[m]$ can hardly be justified outside of this class. For the class of additive valuations, a lower bound of $0.14(1-\frac{1}{n})$ on the quantile of the proportional share follows immediately from Lemma \ref{lem:ber}.

\subsection{Computation} \label{subsec:computation}

Many of the suggested notions of shares (e.g., the maximin share) in the literature are hard to compute and hard to approximate for general monotone valuations (see Appendix~\ref{ap:max-min}). 
In contrast, quantile shares do not suffer from this shortcoming. Indeed, the quantile can be straightforwardly approximated by sampling realizations from the uniformly random allocation (even for the general class of monotone valuations). 

For the classes of additive, unit-demand, and matroid-rank valuations, our proofs suggest an efficient algorithm for computing a $q$-fair allocation for the values of $q$ that admit $q$-fair allocations. 
However, the existence of such a poly-time algorithm for general monotone valuations remains an interesting open problem. 
In particular, our proofs for general monotone valuations are not constructive.

\subsection{Allocation of Bads} \label{subsec:bads}

Fair division has been studied not only for the allocation of goods but also for the allocation of bads (see e.g., \cite{bogomolnaia2017competitive, huang2021algorithmic, aziz2022fair}); namely, the case where $v:2^{[m]}\to \mathbb{R}_{-}$ is monotonically decreasing. We note that the feasibility of the $q$-quantile share for constant $q$ does not extend to the allocation of bads. For example, if a single bad is allocated ($m=1$) then the agent who receives this bad has a quantile of $\frac{1}{n}$. One can easily show using the arguments of Proposition \ref{pro:un} that in this context the critical threshold between feasibility and infeasibility is $q=\frac{1}{n}$.

\bibliographystyle{plainnat}
\bibliography{references}

\appendix

\section{Infeasibility and Computational Hardness of Maximin Fractions}\label{ap:max-min}

\begin{example}
\label{ex:mms-infeasible}
    Consider the case of $n=2$ agents and $m=4$ goods. Let the valuations be
\begin{align*}
    \val_1(X)=\begin{cases}
1 &\text{ if }   X=\{1,2\} \text{ or } X=\{3,4\} \text{ or } |X|\geq 3 \\
0 &\text{ otherwise } \\ 
    \end{cases}\\
    \val_2(X)=\begin{cases}
1 &\text{ if }   X=\{1,3\} \text{ or } X=\{2,4\} \text{ or } |X|\geq 3 \\
0 &\text{ otherwise } \\ 
    \end{cases}
\end{align*}
Note that the maximin share of both agents is 1. Agent 1 can partition the goods into $\{1,2\} \cupdot \{3,4\}$ and agent 2 can partition the goods into $\{1,3\} \cupdot \{2,4\}$. However, every allocation of the goods yields a value of 0 for at least one of the agents; i.e., a 0-fraction of its maximin share.
\end{example}

The above example shows that there exist instances with monotone valuations for which no allocation achieves any positive fraction of the maximin share for all agents. Furthermore, even in a setting with identical (monotone) valuations, where an MMS allocation trivially exists, such an allocation is hard even to approximate. This can be seen, for instance, via a reduction from the NP-complete Partition problem. In this problem, the input is a multiset $S = \{s_1,\ldots,s_m\}$ of positive integers (with $r = s_1+\ldots+s_m$), and the task is to decide whether there exists a partition of $S$ into two submultisets of (equal) sum $r/2$. Given an instance of Partition, consider the valuation function $f: 2^{[m]} \rightarrow \{0,1\}$ constructed as follows: $f(T)=1$ if $\sum_{j\in T} s_j \geq r/2$, and $f(T)=0$ otherwise. Clearly, the function $f$ is monotone and a value oracle can be implemented for it in polynomial time. Now, for the fair division instance with two agents having identical valuation functions $f$, a polynomial-time algorithm that outputs an $\alpha$-MMS allocation for any $\alpha>0$ necessarily finds a partition of $S$ into two multisets of sum $r/2$ if one exists.

\section{Comparison of the Two Bounds of the Erd\H{o}s Matching Conjecture}\label{ap:comp}

\begin{lemma}\label{lem:comp}
    For every $n\geq 2$, $k\geq 1$ and for $m=(k+1)n$ we have 
    $$\binom{m}{k}-\binom{m-n+1}{k}\geq \binom{kn-1}{k}.$$
\end{lemma}

\begin{proof}

After plugging in the value of $m$ and rearranging, the claimed inequality becomes: $$ \frac{\binom{kn-1}{k}}{\binom{kn+n}{k}} + \frac{\binom{kn+1}{k}}{\binom{kn+n}{k}} \leq 1.$$ The first of these ratios is equal to $\prod_{i=0}^n (1-\frac{1}{n+\frac{i}{k}})$ and the second one to $\prod_{i=2}^n (1-\frac{1}{n+\frac{i}{k}})$. As each of the factors in these products is non-increasing in $k$, it suffices to verify that the inequality holds (as an equality) for $k=1$. \end{proof}

\section{Proof of Proposition \ref{lem:n3m6}}\label{ap:m6}

It is sufficient to prove the lemma for $m=6$ because we can just add dummy goods for lower values of $m$. By Proposition~\ref{pro:equiv} we can restrict attention to the case where the agents' valuations are 0/1, and every agent has at most $\frac{4}{9}3^6-1=323$ allocations (among the $3^6=729$) possible with 0 value. We shall prove that there exists an allocation in which every agent gets a 1 value.
We consider two cases.

\paragraph{Case 1: There exists an agent $i$ and a good $j$ such that $v_i(\{j\})=1$.} 
Without loss of generality, we assume $i=3$ and $j=6$.
We consider allocations in which agent 3 gets good 6 only. We consider two subcases.

\paragraph{Case 1.1: There exists an agent $i'\in \{1,2\}$ and a good $j'\in [5]$ such that $v_{i'}(\{j'\})=1$.}
Without loss of generality, we assume $i'=2$ and $j'=5$, and we allocate to agent 2 good 5 only. The remaining goods $[4]$ go to agent 1. We argue that $v_1([4])=1$. Otherwise, agent $1$ gets a $0$ value whenever $S_1 \subseteq [4]$, i.e., in all the $\frac{4}{9}3^6 = 324$ allocations that give goods $5$ and $6$ to agents different from $1$, contradicting our assumption.

\paragraph{Case 1.2: For both agents $i'\in \{1,2\}$ and all goods $j'\in [5]$ we have $v_{i'}(\{j'\})=0$.}

We first argue that agent $2$ has at most $6$ pairs from $[5]$ with 0 value. Otherwise, the number of allocations in which agent 2 gets a 0 value is at least $64+5\cdot 32+7\cdot 16=336$, where $64$ stands for the number of allocations in which agent 2 gets $\emptyset$, $5 \cdot 32$ stands for the number of allocations in which agent 2 gets a singleton from $[5]$, and $7 \cdot 16$ stands for the number of allocations in which agent 2 gets a pair from $[5]$. Since $336>323$ this leads to a contradiction.

Second, we argue that agent $1$ has at most $3$ triples from $[5]$ with 0 value. Assume by way of contradiction that agent $1$ has $4$ triples from $[5]$ with 0 value for her. By the Kruskal-Katona Theorem, since the collection of these triples is of size $4=\binom{4}{3}$, the corresponding collection of pair subsets is of size at least $6=\binom{4}{2}$. Therefore, the number of allocations in which agent 1 gets a 0 value is at least $64+5\cdot 32+6\cdot 16+4\cdot 8=352$, where $64$ stands for the number of allocations in which agent 1 gets $\emptyset$, $5 \cdot 32$ stands for the number of allocations in which agent 1 gets a singleton from $[5]$, $6 \cdot 16$ stands for the number of allocations in which agent 1 gets a pair from $[5]$, and $4 \cdot 8$ stands for the number of allocations in which agent 1 gets a triple from $[5]$. Since $352>323$ this leads to a contradiction.

Among the $\binom{5}{2}=10$ allocations that allocate a pair from $[5]$ to agent 2 and the remaining triple to agent 1, there are at most 6 allocations with 0 value for agent 2 and at most 3 allocations with 0 value for agent 1. We are left with at least one allocation where both agents 1 and 2 have a 1 value.

\paragraph{Case 2: For every agent $i$ and every good $j$ we have $v_i(\{j\})=0$.}
We argue that every agent has at most 4 pairs of goods with 0 value. Otherwise, the number of allocations in which the agent gets a 0 value is at least $64+6\cdot 32+5\cdot 16=336$, where $64$ stands for the number of allocations in which the agent gets $\emptyset$, $6 \cdot 32$ stands for the number of allocations in which the agent gets a singleton, and $5 \cdot 16$ stands for the number of allocations in which the agent gets a pair. Since $336>323$ this leads to a contradiction.

We consider allocations in which every agent gets a pair of goods. We have $\binom{6}{2}\binom{4}{2}=15\binom{4}{2}$ such allocations. Every pair of goods in which an agent gets a 0 value disqualifies $\binom{4}{2}$ of those allocations. So, in total at most $12\cdot \binom{4}{2}$ allocations
are disqualified. We are left with at least $(15-12)\binom{4}{2}>0$ allocations in which all agents have a 1 value.

\end{document}